\documentclass[usenames,dvipsnames,letterpaper,11pt]{article} 
\usepackage{setspace}
\usepackage{array}
\usepackage{enumitem}
\usepackage{amsmath}
\usepackage[T1]{fontenc}
\usepackage[margin=1in]{geometry}
\usepackage[numbers, semicolon]{natbib}
\setcitestyle{authoryear,open={(},close={)},aysep={}} 
\usepackage{amsfonts}
\usepackage{float}
\usepackage{graphicx}
\usepackage{tikz}
\usetikzlibrary{arrows,shapes,snakes,automata,backgrounds,petri,fit, positioning,patterns}
\usepackage{mathrsfs}
\usepackage{caption}
\usepackage{subcaption}
\usepackage[titles,subfigure]{tocloft}
\usepackage{layouts}
\usepackage{hyperref}
\usepackage[normalem]{ulem}
\usepackage{authblk}
\usepackage{bm}
\usepackage{hyperref}

\usepackage[ruled,procnumbered]{algorithm2e}
\usepackage{amsthm}
\usepackage{lscape}
\allowdisplaybreaks

\newtheorem{theorem}{Theorem}
\newtheorem{defi}{Definition}
\newtheorem{example}{Example}

\newtheorem{corollary}{Corollary}
\newtheorem{lemma}{Lemma}

\theoremstyle{remark}
\newtheorem{remark}{Remark}

\usepackage{etoolbox}

\makeatletter
\AtBeginEnvironment{procedure}{\let\c@algocf\c@procedure}
\makeatother

\title{Experimentation on Endogenous Graphs}
\author {
    Wenshuo Wang,
    Edvard Bakhitov,
    Dominic Coey
}
\affil {
    Central Applied Science, Meta\\
    wenshuowang@meta.com, edbakhitov@meta.com, coey@meta.com
}
\date{}

\begin{document}

\maketitle

\begin{abstract}
   We study experimentation under endogenous network interference. Interference patterns are mediated by an endogenous graph, where edges can be formed or eliminated as a result of treatment. We show that conventional estimators are biased in these circumstances, and present a class of unbiased, consistent and asymptotically normal estimators of total treatment effects in the presence of such interference. { We show via simulation that our estimator outperforms existing estimators in the literature.} Our results apply both to bipartite experimentation, in which the units of analysis and measurement differ, and the standard network experimentation case, in which they are the same.
\end{abstract}

\section{Introduction}
\subsection{Background}
In many settings, the fundamental no-interference assumption common in experimental analysis \citep{cox1958planning,rubin1980randomization} is implausible, and there are strong grounds for believing that the treatment assigned to one unit may affect outcomes of other units. This is especially true when outcomes are in part determined by peoples' interactions with each other, as is often the case in public health, education, voting, social media, and other social domains \citep{calvo2009peer,dimaggio2012network, halloran2016dependent,taylor2018randomized}. A large literature has arisen developing the theory of estimation and inference for various causal quantities of interest under these conditions.

Typically, theoretical and empirical analysis proceeds assuming that the interference structure is known—that is, the researcher has prior knowledge of the set of units which might be affected by the treatment of any given unit. This is usually expressed in terms of a known ``exposure mapping'' or ``effective treatment'' function, which determines the level of treatment exposure a unit receives as a function of the treatment assignment vector to all units, or in terms of a ``partial interference'' or ``neighborhood interference'' assumption, according to which interference takes place only within groups, where the groups are known ex-ante \citep{hudgens2008toward,tchetgen2012causal,manski2013identification,liu2014large,aronow2017estimating,forastiere2022estimating}. This is a natural starting point that greatly facilitates the development of treatment effect estimators and the characterization of their behavior.

In practice, the interference structure is rarely known \citep{egami2021spillover,savje2024causal}. Analysts often rely on past user interactions to construct proxy interference graphs \citep{aral2009distinguishing,bakshy2012social,bond201261,coppock2016treatments,harshaw2021design,karrer2021network}. These graphs—such as friendship networks in social media or engagement-based links in marketplaces—are typically crude approximations of the true, unobserved interference graph.\footnote{For example, social network users frequently interact with “unconnected nodes” by being recommended to join new groups or follow new users.} Complicating matters further, they may themselves be causally affected by the treatment of interest. The analyst is presented with an uncomfortable choice—either use a graph constructed only on the basis of pre-treatment data, which may omit relevant information about how treatments affect outcomes, or incorporate post-treatment data too, which can lead to complex biases in treatment effect estimators.

\subsection{Our Contributions}
Motivated by this gap between theory and practice, this paper takes a different approach, explicitly allowing for the observed graph to be determined by treatment assignment. We work in a general bipartite setting, in which units of analysis and randomization are allowed to differ. We propose the notion of an \emph{endogenous} bipartite graph which may be treatment-dependent, and which contains information on interference from randomization to analysis units. Generalizing the outcome model of \citet{harshaw2021design} to this setting, we give an unbiased, consistent and asymptotically normal estimator of the total treatment effect, as well as a consistent testing procedure for the sharp null of no treatment effect. Our identification and estimation approach relies on a treatment-invariant ``anchor'' subgraph, from which instrumental variable estimates of unit-level treatment responses can be obtained.


Our results apply directly to bipartite experiments, which are of growing interest in their own right. Units of analysis and randomization may differ because it is not practical or impossible to measure outcomes of interest at the level at which treatments are being assigned \citep{zigler2021bipartite}. Alternatively, there may be no conceptual difficulty in defining outcomes at the randomization unit-level, but accounting for the bipartite structure is a natural way to allow for cross-unit interference. Causal inference problems with this structure are common in social networks, recommender systems, digital advertising, and multi-sided platforms broadly (\citet{chawla2016b,gilotte2018offline,pouget2019variance,harshaw2021design,nandy2021b,johari2022experimental,bajari2023experimental,shi2024scalable}).


Our framework extends naturally to the standard unipartite setting, where the units of analysis and randomization are the same. 
There, we combine unbiased inverse-propensity weighted estimators of the direct effect, while estimating the indirect effect using the previously developed logic for bipartite estimators. Unbiasedness, consistency and asymptotic normality follow as corollaries to our results for bipartite graphs.

\subsection{Related Work}

Two subareas of network science are especially related to this paper: models of network formation and estimation of peer effects (see \citet{an2011models,chandrasekhar2016econometrics,bramoulle2020peer} for overviews). Particularly relevant is an emerging body of work that questions standard assumptions about interference structure and allows for misspecification of the interference graph \citep{aronow2017estimating,eckles2017design,wang2020design,leung2022causal,savje2024causal}. 

Related work also tackles causal inference on unknown graphs \citep{basse2018limitations,chin2018central,egami2021spillover,savje2021average,cortez2022staggered,yu2022estimating,shirani2023causal,halloran2016dependent}. Like these papers, we are concerned with the implausibility of correctly specifying the interference structure. However, we focus on the case of \textit{endogenous} misspecification: treatment effects propagate along observed edges, but those edges may themselves depend on treatment assignment realizations. Notable recent exceptions that address this setting include \citet{comola2021treatment,gao2024endogenous,ryu2024}. Unlike \citet{comola2021treatment}, we estimate the effect of treating everyone vs no one, the policy-relevant estimand. \citet{ryu2024} takes a non-instrumental variable-based identification approach in the unipartite setting, using conditional means for outcome and edge formation with parameters assumed to be shared across units. Most closely related is \citet{gao2024endogenous}, which also develops an instrumental variable strategy based on a pre-treatment graph, proceeding from quite different modeling assumptions (a unipartite graph with undirected, unweighted edges, and an outcome model with parameters shared across units).

To the best of our knowledge, in addition to providing novel insights into identification and estimation of treatment effects on endogenous unipartite graphs, these are also the first available results on treatment effect estimation for endogenous bipartite graphs. Several important questions remain open, including consistent variance estimation for the total treatment effect estimator.

\subsection{Outline of the Paper} 
The rest of the paper is organized as follows. Section~\ref{sec:setup} introduces the basic setup for endogenous bipartite interference graphs and discusses the edge endogeneity bias. Section~\ref{sec:edge} proposes several models of endogenous edge formation. Section~\ref{sec:tte_estimators} includes our main estimator, and for which we demonstrate unbiasedness, consistency and asymptotic normality under certain assumptions. We also propose a statistical test for the sharp null with asymptotic power equal to $1$. Our theory is supported by simulations in Section~\ref{sec:simulation}, where we examine the biases of our estimator and two exposure reweighted linear estimators \citep{harshaw2021design}, and demonstrate the confidence interval coverage properties of our estimator based on a variance proxy.
Section~\ref{sec:unipartite} extends the results to unipartite graphs, and Section~\ref{sec:conclusion} concludes. All proofs and additional details are included in the Supplementary material.

\section{Setup}
\label{sec:setup}
There are $n_a$ analysis units $a \in \mathcal{A}$ and $n_r$ randomization units $r \in \mathcal{R}$. The treatment assignment random vector is denoted by $\mathbf{T} = (T_1,T_2,\ldots,T_{n_r})$. For each analysis unit $a$, the potential outcome function $Y_a:\{0,1\}^{n_r}  \to \mathbb{R}$ maps treatment assignment vectors to outcomes. Throughout, the object of interest is the total treatment effect (TTE), $\mu = \frac{1}{n_a}\sum_a [Y_a(\mathbf{1}) -  Y_a(\mathbf{0})]$. All randomness comes from treatment assignment. In the usual bipartite experimentation setting, the graph is defined so that an edge exists between an analysis unit and a randomization unit if the potential outcomes of the analysis unit depend on the treatment of the randomization unit. The next definition formalizes this.

\begin{defi} \label{def:bg}
A bipartite interference graph is a triple $(\mathcal{A},\mathcal{R},E)$ of analysis units $\mathcal{A}$, randomization units $\mathcal{R}$, and an adjacency matrix $E$ with representative element $e_{ar} \in \{0,1\}$. For all $a \in \mathcal{A}$, potential outcomes satisfy $Y_a(\mathbf{T}) = Y_a(\mathbf{T'})$ for all $\mathbf{T},\mathbf{T'}$ with $T_r = T_r'$ for all $r \in \mathcal{R}$ such that $e_{ar} = 1$.
\end{defi}
In this definition, the adjacency matrix $E$ restricts how each analysis unit's outcomes can vary in response to treatment assignments, and furthermore is invariant to $\mathbf{T}$. Beyond the trivial example of a fully connected graph, it is rarely plausible to assume that this knowledge is available in applications. We take a different approach, motivated by practical settings in which the analyst observes and wishes to make use of a bipartite graph which contains some important information about how treatments affect outcomes, but cannot plausibly satisfy the definition above. In particular, we allow for the bipartite graph to be determined by the treatment assignment.

\begin{figure}
    \centering
    \includegraphics[width=\linewidth]{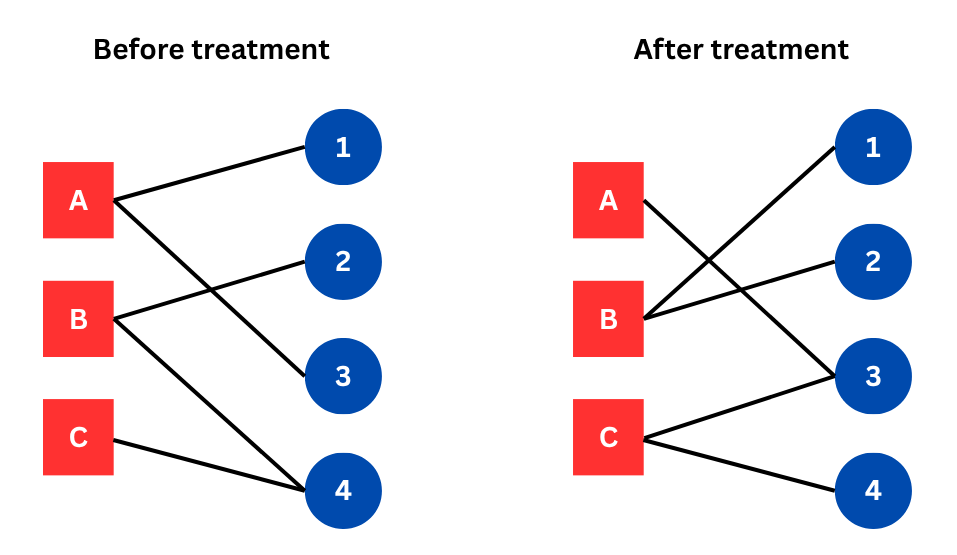}
    \caption{Endogenous graph example. Note that after the treatment is applied, edges (A-1) and (B-4) disappear, while new edges (B-1) and (C-3) emerge.}
    \label{fig:endog_graph}
\end{figure}

Formally, potential outcomes for edges are a function of treatment assignments. Each pair $(a,r) \in \mathcal{A} \times \mathcal{R}$ has a corresponding edge potential outcome function $E_{ar}:\{0,1\}^{n_r} \to \{0,1\} $, i.e.\ specifying whether the corresponding edge exists as a function of the treatment assignment vector. Let $\mathcal{R}_a(\mathbf{T}) = \{r \in \mathcal{R}:E_{ar}(\mathbf{T}) = 1\}$ denote the randomization units $a$ is connected to under the treatment $\mathbf{T}$. We assume that as long as a treatment assignment does not change the randomization units that analysis unit $a$ is connected to and does not change the treatment assignments to those randomization units, it cannot change $a$'s outcome. This gives the definition of an endogenous bipartite interference graph. 

\begin{defi} \label{def:ebg}
An endogenous bipartite interference graph is a triple $(\mathcal{A},\mathcal{R},E)$ of analysis units $\mathcal{A}$, randomization units $\mathcal{R}$, and an adjacency matrix $E$ with representative element $E_{ar} : \{0,1\}^{n_r} \rightarrow \{0,1\}$. For all $a \in \mathcal{A}$, potential outcomes satisfy $Y_a(\mathbf{T}) = Y_a(\mathbf{T'})$ for all $\mathbf{T},\mathbf{T'}$ with $\mathcal{R}_a(\mathbf{T}) = \mathcal{R}_a(\mathbf{T'})$ and $T_r = T_r'$ for $r \in \mathcal{R}_a(\mathbf{T})$.
\end{defi}

We assume throughout that for all $(a,r) \in \mathcal{A} \times \mathcal{R}$ the functions $E_{ar}(\cdot)$ are unknown, but their realizations given the treatment assignment $E_{ar}(\mathbf{T})$ are observed (see Figure \ref{fig:endog_graph} for illustration). Some observations on the relation between bipartite vs.\ unipartite and endogenous vs.\ exogenous interference graphs follow.

\begin{remark}
Definition~\ref{def:ebg} generalizes Definition~\ref{def:bg}, with the latter corresponding to the special case of the former where $\mathcal{R}_a(\mathbf{T})$ is constant.
\end{remark}

\begin{remark}
Definition~\ref{def:ebg} also encompasses the case of endogenous unipartite experimentation, which is the special case of $\mathcal{A} = \mathcal{R}$. In this case, we typically assume $E_{aa} = 1$ for all $a$, so that each unit's treatment assignment is allowed to affect itself. The interference graph is allowed to be directed, in the sense that $E_{aa'}$ need not necessarily equal $E_{a'a}$ for $a \ne a'$.
\end{remark}

\begin{remark}
Given any endogenous bipartite interference graph $(\mathcal{A},\mathcal{R},E)$, it is always possible to find \emph{some} (exogenous) bipartite interference graph $(\mathcal{A},\mathcal{R},E')$ consistent with $(\mathcal{A},\mathcal{R},E)$. Trivially, we can take $E'$ to be the fully connected bipartite graph $\mathcal{A} \times \mathcal{R}$. Somewhat less trivially, we can define $E'$ as the set of edges $(a,r)$ such that there exists a treatment realization $\mathbf{T}$ satisfying $E_{ar}(\mathbf{T}) = 1$. In principle, conventional methods for causal inference on networks could be applied on the graph $(\mathcal{A},\mathcal{R},E')$. This observation is of limited value to practitioners, if $E'$ is unobserved and they only observe a particular, treatment-dependent realization $E(\mathbf{T})$ of endogenous edges.
\end{remark}


Graphs built from behavioral signals may not satisfy Definition~\ref{def:bg} if they are likely to be causally affected by treatments, but they may still satisfy Definition~\ref{def:ebg}. For instance, a user-item graph in an online marketplace where edges represent shown items, may be causally affected by item-level treatments. Yet we might still assume that user outcomes remain unchanged if neither the set of shown items nor their treatments change. Similarly, in a Facebook Groups experiment, user outcomes may plausibly be unaffected if a user's group memberships and the treatments of those groups remain the same.

In what follows, we assume we have access to a pre-treatment graph, constructed from data before the experiment starts. It is denoted $G^{pre} = (\mathcal{A},\mathcal{R},E^{pre})$, with representative element $e_{ar}^\text{pre} \in \{0,1\}$.\footnote{Even though no-one is treated pre-experiment, we need not have $E^{pre} = E(\mathbf{0})$. The potential outcomes for the edges and the metric of interest, $(E(\mathbf{T}),Y_a(\mathbf{T}))$ should be interpreted as their values at some fixed time post-experiment start. So if edges form and disappear over time regardless of whether treatments are being applied, the pre-experiment edges $E^{pre}$ and the counterfactual edges that would exist without any treatment as measured at that future time, $E(\mathbf{0})$, may differ.} For conciseness, unless we wish to make their dependence on $\mathbf{T}$ explicit, we write the realized metric outcomes $Y_a(\mathbf{T})$ as $y_a$ and realized edge outcomes $E_{ar}(\mathbf{T})$ as $e_{ar}$. We now give examples of how naively applying existing bipartite estimators without accounting for the endogenous nature of the graph can bias conventional inference-aware estimators.

\paragraph{Examples of Edge Endogeneity Bias.}
\label{sec:bias}

Treating the realized graph $E(\mathbf{T})$ observed in the setting of Definition \ref{def:ebg} as fixed can bias otherwise unbiased estimators, including the Horvitz-Thompson (HT) inverse propensity weighted estimator (\citet{horvitz1952generalization}). We assume that treatment assignments are independent Bernoulli random variables with probability $p$. The HT estimator of the TTE is  

\begin{align*}
    \hat{\mu}_\text{HT} &= \frac{1}{n_a} \sum_{a = 1}^{n_a} y_a \left[ \frac{\prod_{r \in \mathcal{R}_a(\mathbf{T})} T_r }{p^{|\mathcal{R}_a(\mathbf{T})|}} - \frac{\prod_{r \in \mathcal{R}_a(\mathbf{T})} (1 - T_r) }{(1 - p)^{|\mathcal{R}_a(\mathbf{T})|}} \right],
\end{align*}
where the term in the square brackets is defined to be zero if $\mathcal{R}_a(\mathbf{T}) = \emptyset$.
The first term simply sums up outcomes $y_a$ for the ``fully treated'' analysis units, i.e.\ those units $a$ such that all their corresponding randomization units are treated ($\prod_{r \in \mathcal{R}_a(\mathbf{T})} T_r = 1$), and rescales by the factor $p^{-|\mathcal{R}_a(\mathbf{T})|}$ to account for the probability of that event. This gives the estimate of the average outcomes under $\mathbf{T} = \mathbf{1}$, and analogously the second term gives the estimate of average outcomes under $\mathbf{T} = \mathbf{0}$. By standard arguments \cite{aronow2017estimating}, this is an unbiased estimate of the TTE under graph exogeneity, where $\mathcal{R}_a(\mathbf{T})$ does not vary with $\mathbf{T}$.

Under graph endogeneity, this unbiasedness property no longer holds. We illustrate mechanisms by which bias can arise in the following minimal examples. There is a single randomization unit in each, implying the HT estimator and the exposure reweighted linear (ERL) estimator of \citet{harshaw2021design} are identical, so these examples also demonstrate bias in the ERL estimator.

\begin{example}\label{exmp1}
$\mathcal{R} = \{r_1\}$ and $\mathcal{A} = \{a_1\}$. Unit $r_1$ forms an edge with $a_1$ if and only if $T_1 = 1$. Outcomes for $a_1$ are $y\ne0$ regardless of the treatment assignment, and $T_1\sim\text{Bernoulli}(1/2)$. The HT estimator is $\frac{T_1 y}{1/2}$. This is $y$ in expectation, which is biased for the TTE of zero. In this example the number of edges varies with the treatment assignment, but that is not necessary for bias to exist, as the next example shows.
\end{example}

\begin{example}\label{exmp2}
$\mathcal{R} = \{r_1\}$ and $\mathcal{A} = \{a_1, a_2\}$. Unit $r_1$ forms an edge with $a_1$ if $T_1 = 1$, and with $a_2$ otherwise. Outcomes for $a_1$ and $a_2$ are $y_1$ and $y_2$ regardless of the treatment assignment, and $T_1\sim\text{Bernoulli}(1/2)$. The HT estimator is $\frac{T_1 y_1}{1/2} - \frac{(1 - T_1) y_2}{1/2}$ which has expectation $y_1 - y_2$, and again is biased for the TTE of zero. In this example, edge formation induces correlation between treatment intensity and analysis unit potential outcomes, generating bias.
\end{example}

The preceding examples distinguish between randomization and analysis units and so relate to bipartite experimentation. Graph endogeneity can also bias treatment effect estimators in the case of unipartite network interference where this distinction does not exist, as the next example shows.

\begin{example}\label{exmp3}
$\mathcal{R} = \mathcal{A} = \{a_1, a_2\}$. Outcomes for $a_1$ and $a_2$ are $y_1$ and $y_2$ regardless of treatment assignment, and $T_1$, $T_2$ are independent $\text{Bernoulli}(1/2)$ draws. A directed edge runs from $a_1$ to $a_2$ if and only if $T_1 = 1$. The HT estimator is 
\begin{align*}
\hat{\mu}_\text{HT} = \begin{cases}
-y_1 / 2 - y_2 / 2 &\text{if $T_1 = T_2 = 0$}\\
-y_1 / 2 + y_2 / 2 &\text{if $T_1 = 0, T_2 = 1$}\\
y_1 / 2 &\text{if $T_1 = 1, T_2 = 0$}\\
y_1 / 2 + y_2 / 4 &\text{if $T_1 = T_2 = 1$},\\
\end{cases}
\end{align*}
with expectation $y_2 / 16$, and is biased for the TTE of zero.
\end{example}

\section{Models of Edge Formation}
\label{sec:edge}
Just as estimating total treatment effects under exogenous interference requires imposing some additional structure on the interference mechanism \citep{basse2018limitations}, we will need to rule out arbitrary dependence of edges on treatment assignments. Let $S_{ar}$ be a subset of $\mathcal{R}$ satisfying $E_{ar}(\mathbf{T}) = E_{ar}(\mathbf{T'})$ for all $\mathbf{T}, \mathbf{T}'$ such that $T_s = T_s'$ for all $s \in S_{ar}$. That is, whether the edge $(a,r)$ exists depends only on the treatment assignments to the randomization units in $S_{ar}$. Examples of various edge formation assumptions follow.

\begin{itemize}   
    \item \emph{Unrestricted edges}. $S_{ar} = \mathcal{R}$. Edges can depend arbitrarily on treatments.
    \item \emph{Exogenous edges}. $S_{ar} = \emptyset$. Edges are unaffected by treatments, as per Definition \ref{def:bg}.    \item \emph{r-driven edges}. $S_{ar} = \{r\}$. Edges depend only on the randomization unit's treatment status. This can be interpreted as a SUTVA (Stable Unit Treatment Value Assumption) assumption on edge formation: edge outcomes depend only on the treatment assignment to the randomization unit in question, not by treatment assignments to other units.
    \item \emph{$S_a$-driven edges}. $S_{ar} = S_a \subset \mathcal{R}$. Analysis unit $a's$ edge formation depends only on treatment assignments to a subset of randomization units.
    \item \emph{$S_r$-driven edges}. $S_{ar} = S_r \subset \mathcal{R}$. Randomization unit $r's$ edge formation depends only on treatment assignments to a subset of randomization units.
\end{itemize}


The following statistical tests can provide evidence on how appropriate the assumptions of exogenous or $r$-driven edges may be.

\paragraph{Detecting general departures from edge exogeneity.}

Define the null hypothesis as $H_0: S_{ar}=\emptyset$ for all $a,r$, and
consider the following test statistic based on the binomial likelihood for each analysis unit's number of treated edges:
\begin{equation*}
W(\mathbf T, E(\mathbf T))=\sum_a \log \left[{\sum_r e_{ar} \choose \sum_r e_{ar} T_r}p^{\sum_r e_{ar} T_r} (1 - p)^{\sum_r e_{ar} (1 - T_r)} \right].
\end{equation*}
Under $H_0$, $E(\mathbf T)$ is independent of $\mathbf T$, and the conditional distribution of $W$ given $E$ can be simulated by resampling $\mathbf T$ according to the experiment design. Following the usual randomization inference logic, we obtain a valid statistical test by rejecting the null if the value of $W$ is larger than the upper $\alpha^{\textrm{th}}$ quantile of this null conditional distribution.

This test has power against a variety of departures from exogenous edges, including but not limited to $r$-driven edges. If treated randomization units disproportionately form new edges, this will lead to analysis units having more treated units than would be expected compared to the binomial benchmark. And if analysis units are more likely to form additional edges based on the treatment intensity of their existing edges, this will similarly lead to a departure from the binomial benchmark.\footnote{To build intuition for the former case, imagine units form edges if and only if they are treated. Then for all analysis units that have any edges, they are all maximally treated. For the latter case, imagine analysis units will form an edge with a single predetermined randomization unit, and then form another edge with another randomization unit if and only if the previous unit was treated. Then there will be no analysis units connected to a single treated randomization unit.} It is however not powered to detect all conceivable departures from exogenous edges, and there are somewhat contrived examples of endogenous edges which will not be detected by this test.\footnote{Consider the case where there is a single analysis unit and two randomization units, and $e_{11} = \mathbf{1}(T_2 = 1)$. The treatment intensity for the analysis unit conditional on an edge existing is indeed binomial, despite the fact that whether an edge exists in the first place depends on the other randomization unit.}

\paragraph{Detecting $r$-driven edges.}
If the alternative hypothesis is that edges only depend on the randomization unit's treatment status, a straightforward and valid procedure is to test whether treatments are generating more or fewer edges in aggregate with a simple difference-in-means \emph{t}-test. Let $E_r=\sum_a(e_{ar} -e_{ar}^\text{pre})$ be the net number of edges added relative to the pre-experiment graph. We can then test the null hypothesis that $\mathbb E[E_r\mid T_r=1]=\mathbb E[E_r\mid T_r=0]$.

For TTE estimation, we focus exclusively on the case of $r$-driven edges. This strikes a balance between meaningfully enriching the standard exogenous bipartite graph model, and still imposing enough structure on the problem to be able to derive useful statements on treatment effect magnitudes. This choice is also motivated by practicality—while setting $S_{ar}$ to be some strict superset of $\{r\}$ is more general, in applications there may be no clear principles available for constructing such sets. Note that we are relaxing the common assumption of the interference graph being fixed, and in this respect generalize, rather than specialize, existing work---the $r$-driven edges assumption is trivially satisfied for a fixed graph. Some restrictions on edges and outcomes are indeed necessary to make meaningful claims about treatment effects in this more general setting.

We opt for the $r$-driven model of edge endogeneity as the simplest non-trivial generalization of exogenous edges. Robustness under the violation of this assumption and relaxing the assumption to $S_r$-driven edges, where edge formation can depend on peers’ treatments, are interesting future directions to explore. We wish to note that, while it is attractive to dispense with this assumption entirely, we see it as being analogous to the required restrictions on the interference structure in the exogenous graph case—without some structure, causal inference is impossible (see, e.g., \citet{basse2018limitations}). 

\section{A Class of Unbiased TTE Estimators}
\label{sec:tte_estimators}

\subsection{Generalized Linear Exposure-Response Model}
We specify the outcome model $y_a = \alpha_a + \beta_a x_a$, where $x_a = \sum_r T_r e_{ar}w_{ar}$ is the weighted number of treated edges $a$ is connected to in the realized, post-treatment graph, and $w_{ar}$'s are known weights. The $w_{ar}$'s are a modeling choice, and we recommend the default choice of
\begin{equation*}
w_{ar}=1/(\text{\# pre-treatment edges connected to }a).
\end{equation*}
The researcher might choose to incorporate domain knowledge into these weights based on how they expect a node to be affected by others (e.g., a user may be more affected by friends with whom they message frequently, or a buyer may be more affected by a seller they have frequently bought from in the past). These weights play the same role as the weights in the outcome model of \citet{harshaw2021design}. This outcome model satisfies the restrictions on potential outcomes stipulated in Definition~\ref{def:ebg}, and is a generalization of the linear exposure-response model in \citet{harshaw2021design} to endogenous bipartite graphs. Outcomes are linear in the number of treated randomization units the analysis unit is exposed to, but the linear response is allowed to vary across analysis units. Treatments only affect outcomes through the observed edges, but the set of observed edges is potentially treatment dependent. By definition, the TTE is 
\begin{equation}
\mu =  \frac{1}{n_a}\sum_a [Y_a(\mathbf{1}) - Y_a(\mathbf{0})]=\frac{1}{n_a}\sum_a\mathcal W_a(\mathbf 1)\beta_a,
\label{equation:gate}
\end{equation}
where $\mathcal W_a(\mathbf 1)=\sum_rw_{ar}E_{ar}(\mathbf 1) $.

\subsection{Estimation Strategy}
We now give a high-level overview of our approach to estimating the TTE. The two components of \eqref{equation:gate} are $\beta_a$, the incremental effect of treating an additional randomization unit for analysis unit $a$, and $\mathcal W_a(\mathbf 1)$, the total number of edges $a$ would have under a full treatment roll-out. To construct an estimator of the TTE, we propose a class of unbiased estimators for each of these two components. We prove that the estimators for these two components are in addition uncorrelated with each other, so their product is an unbiased estimator for the summand  of \eqref{equation:gate}.

Considering $\beta_a$ first, \citet{harshaw2021design} propose $y_a(x_a - \mathbb{E}{x_a})/{\textrm{Var}(x_a)}$ as an unbiased estimator in the exogenous bipartite graph case. This estimator is no longer feasible because the expectation and variance of $x_a$ are unknown, due to the endogenous nature of the graph edges. Moreover, as the examples in Section~\ref{sec:bias} show, naively applying this estimator assuming the realized graph is exogenous will lead to bias. To remedy this issue, we borrow from the literature on instrumental variables. Let $z_a^u = \sum_r T_r u_{ar}$ be a weighted exposure variable, where $u_{ar}$'s are fixed weights chosen by the researcher. If $\textrm{Cov}(x_a, z^u_a)$ is known and non-zero, we have a class of unbiased estimators for $\beta_a$,
\begin{equation}
\left\{\hat\beta_a^u=y_a(z^u_a - \mathbb{E}{z^u_a})/\textrm{Cov}(x_a, z^u_a)\mid u:\mathcal A\times\mathcal R\to\mathbb R\right\}.
\label{equation:beta_hat}
\end{equation}

For $\mathcal W_a(\mathbf 1)$, let
$$ \widehat{\mathcal W_{ar}(\mathbf 1)}^c = \frac{T_r w_{ar}(e_{ar} - c_{ar})}{p}+w_{ar}c_{ar}.$$
Under $r$-driven edges it follows that $\mathbb{E}[\widehat{\mathcal W_{ar}(\mathbf 1)}^c] = w_{ar}\mathbb{E}[e_{ar} \mid T_r = 1] = w_{ar}E_{ar}(\mathbf{1})$, and so a class of unbiased augmented inverse probability weighted estimators is
\begin{equation}
\left\{\widehat{\mathcal W_a(\mathbf 1)}^c=\sum_r \widehat{\mathcal W_{ar}(\mathbf 1)}^c 
 \mid c:\mathcal A\times\mathcal R\to\mathbb R\right\}.
\label{equation:w_hat}
\end{equation}

Our estimation strategy involves giving conditions ensuring that $\textrm{Cov}(x_a, z^u_a)$ is known and non-zero, and that our estimator for $\beta_a$ is uncorrelated with our estimator for $\mathcal W_a(\mathbf 1)$. The key assumption is the existence of an observed ``anchor'' subgraph, the edges of which would exist under any treatment. This suffices for unbiasedness of our TTE estimator. Consistency and asymptotic normality follow under additional assumptions on graph sparsity, which limit dependence across analysis units and allow the application of convergence theorems for dependent data.

\subsection{Using Anchor Instruments}
\label{sec:gate-anchor}



In general, though the experimental design and hence the joint distribution of $\mathbf{T}$ is known, the covariance between the treatment intensity $x_a = \sum_r T_r e_{ar}w_{ar}$ and the instrument $z_a^u = \sum_r T_r u_{ar}$ is unknown. This is because it is a function of the unobserved values that the random variables $e_{ar}$ would take on under counterfactual treatment assignments. This suggests a path forward for a particular class of instruments for which $\textrm{Cov}(x_a, z_a^u)$ can be calculated. If there were a known set of pre-experiment edges that would continue to exist regardless of treatment assignments, this would allow us to circumvent the difficulties posed by not observing $e_{ar}$ under counterfactual treatments. To start formalizing these observations, we present the definition of an anchor subgraph.

\begin{defi}[Anchor Subgraph]
In Definition~\ref{def:ebg}, we say $G\subset\mathcal A\times\mathcal R$ is an anchor subgraph, if $E_{ar}(\mathbf T)=1$ for any $\mathbf T\in\{0,1\}^{n_r}$, $(a,r)\in G$.
\end{defi}

An anchor subgraph is a set of edges that would exist under any treatment.\footnote{For our Theorem~\ref{theorem:unbiased} and Theorem~\ref{theorem:consistent}, we can relax this assumption to requiring only $E_{ar}(\mathbf{1})=1$ under $r$-driven edges. Thus the existence of a known set of edges in the pre-experiment graph which satisfy the anchor subgraph property can be seen as a milder version of a ``compliance''-type assumption familiar from the instrumental variables literature \citep{angrist1996identification}, which in this setting would amount to requiring that all pre-treatment edges exist after the treatment has been applied.}  A subgraph in the pre-experiment data may plausibly satisfy this property, in which case we can use it to construct an instrument for the endogenous $x_a$'s.


As a practical example, consider Facebook groups, where the bipartite graph links users (analysis units) to groups (randomization units). User outcomes depend on the treatment of connected groups, and those connections may themselves be affected by the treatment. To construct an anchor subgraph, we can identify groups a user is unlikely to leave—e.g., using a predictive model trained on pre-treatment data. Similarly, in unipartite settings, persistent friendships may serve as anchor edges. These assumptions can be validated empirically—for any candidate anchor subgraph, we can estimate what fraction of edges disappear under treatment. To continue the Facebook groups example, less than $0.5\%$ of all group memberships disappear over two weeks. By contrast, experimental treatments may cause a larger number of new memberships to form over the same period, with a recent test generating a rate of edge formation of 0.75\%.

The following theorem gives conditions for the TTE estimator we propose to be unbiased. Assumption (a) of $r$-driven edges restricts how treatment assignment affects edge creation or deletion and is used throughout the proof. Assumption (b) describes the experimental design. Assumption (c) requires that the variation in the instrument $z_a^{u}$ is entirely driven by treatment assignments to the known anchor subgraph $G$. Assumption (d) ensures $\textrm{Cov}(\hat\beta_a^u,\widehat{\mathcal W_a(\mathbf 1)}^c) = 0$. Assumptions (e) and (f) guarantee that the instrument $z_a^{u}$ has some non-zero correlation with the treatment intensity $x_a$. 

\begin{theorem}
In the setting of Definition~\ref{def:ebg}, suppose $G$ is a known anchor subgraph. If (a) edges are $r$-driven; (b) the treatment assignments are independent Bernoulli random variables with probability $p$; (c) $\{(a,r)\mid u_{ar}\ne0\}\subset G$; (d) $c_{ar}=\mathbb I((a,r)\in G)$; (e) $w_{ar}\ne0$ for all $(a,r) \in G$; and (f) $|\{r\mid u_{ar}\ne0\}|>0$ for all $a$, then
\begin{equation}
\hat\mu^{u,c}=\frac{1}{n_a}\sum_a\hat\beta^u_a\cdot\widehat{\mathcal W_a(\mathbf 1)}^c
\label{equation:mu-hat}
\end{equation}
is an unbiased estimator of the TTE in \eqref{equation:gate}.
\label{theorem:unbiased}
\end{theorem}

Let $R_a=\{r\in\mathcal R\mid E_{ar}(\mathbf{1})=1\}$ be the set of randomization units connected to $a$ under treatment, and $A_r=\{a\in\mathcal A\mid E_{ar}(\mathbf{1})=1\}$ the set of analysis units connected to $r$ under treatment. Let $V_a=\{r\in\mathcal R\mid u_{ar}\ne0\}$ be the set of randomization units connected to $a$ in the anchor subgraph. To see the time complexity of computing \eqref{equation:mu-hat}, we note that
\begin{equation*}
\hat\beta^u_a=\frac{y_a\sum_ru_{ar}(T_r-p)}{p(1-p)\sum_rw_{ar}u_{ar}}
\end{equation*}
requires $\Theta(|V_a|)$ calculations,
and $\widehat{\mathcal W_a(\mathbf 1)}^c$ requires $\Theta(|R_a|)$ calculations. Thus, the total time complexity is $\Theta\left(\sum_a|R_a|\right)$.

To characterize the limiting behavior of the estimator proposed in Theorem~\ref{theorem:unbiased}, we consider an asymptotic regime in which the bipartite graph increases in size. If the graph is sufficiently sparse in this limit—i.e.\ the maximum degree of the analysis units and randomization units grows sufficiently slowly relative to the number of analysis units—then we can apply convergence results for dependent data, which allow us to derive consistency and asymptotic normality. Asymptotic normality relies on an application of Stein's method \citep{stein1986approximate,ross2011fundamentals}, which has proven to be broadly applicable in demonstrating asymptotic normality in network settings \citep{chin2018central,harshaw2021design}. Theorem \ref{theorem:consistent}
below summarizes these results. Assumptions (a), (b) and (c) ensure the instrument is positively correlated with the treatment exposure level $x_a$. Assumption (c) limits the total weight of randomization units connected to the same user and the discrepancy between those weights, and assumption (d) bounds outcomes.

\begin{theorem}
We maintain the assumptions in Theorem~\ref{theorem:unbiased}, and further assume (a) $|V_a|>0$; (b) $u_{ar}\ge0$; (c) $W_l/|R_a|\le|w_{ar}|\le W_h/|R_a|$ for $r\in R_a$; and (d) $|y_a|\le M$, where $p$, $1-p$, $W_l$, $W_h$, $M$ are constants bounded away from $0$ and $\infty$. Then, letting $d_{\mathcal A}=\max_{a\in \mathcal A}|R_a|$ and $d_{\mathcal R}=\max_{r\in \mathcal R}|A_r|$, \begin{enumerate}
    \item $\text{Var}(\hat\mu^{u,c})=O(d_{\mathcal A}^3d_{\mathcal R}/n_a)$, and thus $\hat\mu^{u,c}$ is consistent if $d_{\mathcal A}^3d_{\mathcal R}/n_a\to0$;
    \item if $n_a\text{Var}(\hat\mu^{u,c})$ is bounded away from zero and $d_{\mathcal A}^{10}d_{\mathcal R}^4/n_a\to0$, then ${(\hat\mu^{u,c}-\mu})/{\sqrt{\text{Var}(\hat\mu^{u,c})}}$ converges in distribution to $\mathcal N(0,1)$.
\end{enumerate}
\label{theorem:consistent}
\end{theorem}

With asymptotic normality and a variance estimator, we can construct confidence intervals. Section~\ref{subsec:testing} provides a modified estimator with tractable variance that can be used as a proxy for uncertainty quantification for the original estimator. Section~\ref{sec:simulation} shows this proxy variance yields confidence intervals with close-to-nominal coverage in simulations.

\subsection{Statistical Testing under the Sharp Null}
\label{subsec:testing}
Based on the estimators of \eqref{equation:beta_hat}, we can construct a test for the sharp null $H_0^\text{sharp}$: $\beta_a=0$ for all $a$. Under $H_0^\text{sharp}$, $y_a=\alpha_a$ is independent of $\mathbf T$. Define a modified variant of \eqref{equation:mu-hat},
\begin{equation}
\begin{aligned}
\tilde\mu^{u,c}&=\frac1{n_a}\sum_a\hat\beta_a^u\sum_r(w_{ar}c_{ar})\\
&=\frac1{n_a}\sum_a\frac{y_a(z_a^u(\mathbf T)-\mathbb Ez_a^u)}{\text{Cov}(x_a,z_a^u)}\sum_r(w_{ar}c_{ar}),
\end{aligned}
\label{equation:mu-tilde-for-test}
\end{equation}
where the dependence of $z_a^u$ on $\mathbf T$ is explicit. Under $H_0^\text{sharp}$ and the assumptions of Theorem~\ref{theorem:unbiased}, all randomness in $\tilde\mu^{u,c}$ is from $z_a^u(\mathbf T)$, the distribution of which is known.
\begin{theorem}
We maintain the assumptions in Theorem~\ref{theorem:consistent}. Define the sharp null $H_0^\textnormal{sharp}$: $\beta_a=0$ for all $a$. The $H_0^\textnormal{sharp}$-null distribution of $|\tilde\mu^{u,c}|$ in \eqref{equation:mu-tilde-for-test} conditional on $\{y_a: a\in\mathcal A\}$ is known. Denote $c_\alpha$ as its $\alpha$-upper quantile, $\alpha\in(0,1)$. Rejecting if $|\tilde\mu^{u,c}| \ge c_\alpha$ is a valid size-$\alpha$ test for $H_0^\textnormal{sharp}$, and the power goes to $1$ if $d_{\mathcal A}^3d_{\mathcal R}/n_a\to0$ and 
\begin{equation*}
    \liminf\left|\frac1{n_a}\sum_a\beta_a\sum_r(w_{ar}c_{ar})\right|>0.
\end{equation*}
\label{theorem:test}
\end{theorem}

The time complexity to calculate \eqref{equation:mu-tilde-for-test} is $\Theta\left(\sum_a|V_a|\right)$. The exact $c_\alpha$ is expensive to compute, but noting that $\text{E}[\tilde\mu^{u,c}\mid \{y_a: a\in\mathcal A\}, H_0^\text{sharp}]=0$, we can obtain a more conservative cutoff based on $\text{Var}(\tilde\mu^{u,c}\mid \{y_a: a\in\mathcal A\}, H_0^\text{sharp})$, the time complexity to calculate which is also $\Theta\left(\sum_a|V_a|\right)$, because we have
\begin{equation*}
\text{Var}(\tilde\mu^{u,c}\mid \{y_a: a\in\mathcal A\}, H_0^\text{sharp})
=\frac{1}{p(1-p)}\sum_r\left(\sum_a\frac{y_aC_au_{ar}}{n_aU_a}\right)^2,
\end{equation*}
where $C_a=\sum_rw_{ar}c_{ar}$ and $U_a=\sum_rw_{ar}u_{ar}$.

\section{Simulations}
\label{sec:simulation}
\subsection{Simulation Setting}

We let $n_a=2000$ and $n_r=400$. We generate pre-treatment edges $e_{ar}^\text{pre}$'s as independent $\text{Bern}(0.1)$ random variables. Our algorithm takes the anchor edges to be the same as the pre-treatment edges. The edge outcome function is $r$-driven and defined as 
$E_{ar}(T_r)= e_{ar}^\text{pre}+T_r(1-2e_{ar}^\text{pre})\text{Bern}(p_\text{anchor})$ if $(a,r)$ is anchor, and $e_{ar}^\text{pre}+T_r(1-2e_{ar}^\text{pre})\text{Bern}(p_\text{non-anchor})$ otherwise.
That is, the edge outcome is the same as the pre-treatment edge if no treatment is applied to the randomization unit; otherwise, it will flip to the opposite of the pre-treatment edge with probability $p_\text{anchor}$ if the pre-treatment edge exists, and $p_\text{non-anchor}$ if it does not.

The $\alpha_a$'s are sampled from independent $\text{Unif}([0, 1])$ variables, and the true treatment effect is set to be $\beta_a=\beta$ for all $a\in\mathcal A$. We run 1000 Monte Carlo simulations for various values of $\beta$, $p_\text{anchor}$ and $p_\text{non-anchor}$ and compare our anchor-based estimator (Anchor) with the ERL estimator \citep{harshaw2021design} based on (a) pre-treatment graph ($\text{ERL}_\text{pre}$) and (b) post-treatment graph ($\text{ERL}_\text{post}$). 
The ERL estimators are
\begin{equation*}
\hat\mu^{\text{ERL}_\text{pre}}=\frac1{n_a}\sum_a\frac{y_a(x_a^\text{pre}(\mathbf T)-p)}{p(1-p)/(\sum_{r\in\mathcal R}e_{ar}^\text{pre})}
\end{equation*}
and
\begin{equation*}
\hat\mu^{\text{ERL}_\text{post}}=\frac1{n_a}\sum_a\frac{y_a(x_a^\text{post}(\mathbf T)-p)}{p(1-p)/(\sum_{r\in\mathcal R}e_{ar}^\text{post})},
\end{equation*}
where $x_a^\text{pre}(\mathbf T)=\frac{\sum_{r\in\mathcal R}e_{ar}^\text{pre}T_r}{\sum_{r\in\mathcal R}e_{ar}^\text{pre}}$ and $x_a^\text{post}(\mathbf T)=\frac{\sum_{r\in\mathcal R}e_{ar}^\text{post}T_r}{\sum_{r\in\mathcal R}e_{ar}^\text{post}}$.
Here, $e_{ar}^\text{pre}$ and $e_{ar}^\text{post}$ are the pre-treatment and post-treatment edges between $a$ and $r$, respectively. We set $\beta=2$, $p_\text{anchor}=2\%$ and $p_\text{non-anchor}=0.2\%$ when they are not being varied.

\subsection{Bias Comparison}

Figure~\ref{figure:bias} shows the biases of the estimators under different parameter values. The anchor estimator results are consistent with unbiasedness with respect to different values of treatment effect ($\beta$) and edge formation intensity ($p_\text{non-anchor}$), as expected given our theoretical results. Using the post-treatment graph ($\text{ERL}_\text{post}$) naively can lead to severe overestimation due to the positive correlation between the marginal effect of an edge and the number of treated edges. Estimates based on the pre-treatment graph are still biased, but in the opposite direction, and the magnitude of the bias is much smaller as it stems solely from graph misspecification. Unbiasedness of the anchor estimator is achieved at the cost of higher variance, and the variance grows with $p_\text{non-anchor}$; with more edges being formed in response to treatment, we expect larger metric fluctuations between test and control contributed by $\widehat{\mathcal W_{ar}(\mathbf 1)}^c$, and hence, higher variance of the anchor estimator.

The anchor estimator is quite robust to violations of the anchor subgraph assumption, as we increase $p_\text{anchor}$. In Figure~\ref{figure:bias}, we can see that its bias increases with $p_\text{anchor}$, but stays relatively small compared to the magnitude of the treatment effect, which is $2$. $\text{ERL}_\text{pre}$ is not directly affected by $p_\text{anchor}$ and is stable. On the other hand, the bias of $\text{ERL}_\text{post}$ changes its sign with the number of anchor edges disappearing. This is most likely due to the correlation between the marginal effect (which is positive) and the number of edges (that starts to diminish) becoming negative leading to underestimation of the true effect.


\begin{figure*}[h]
  \centering
  \includegraphics[width=0.61\linewidth]{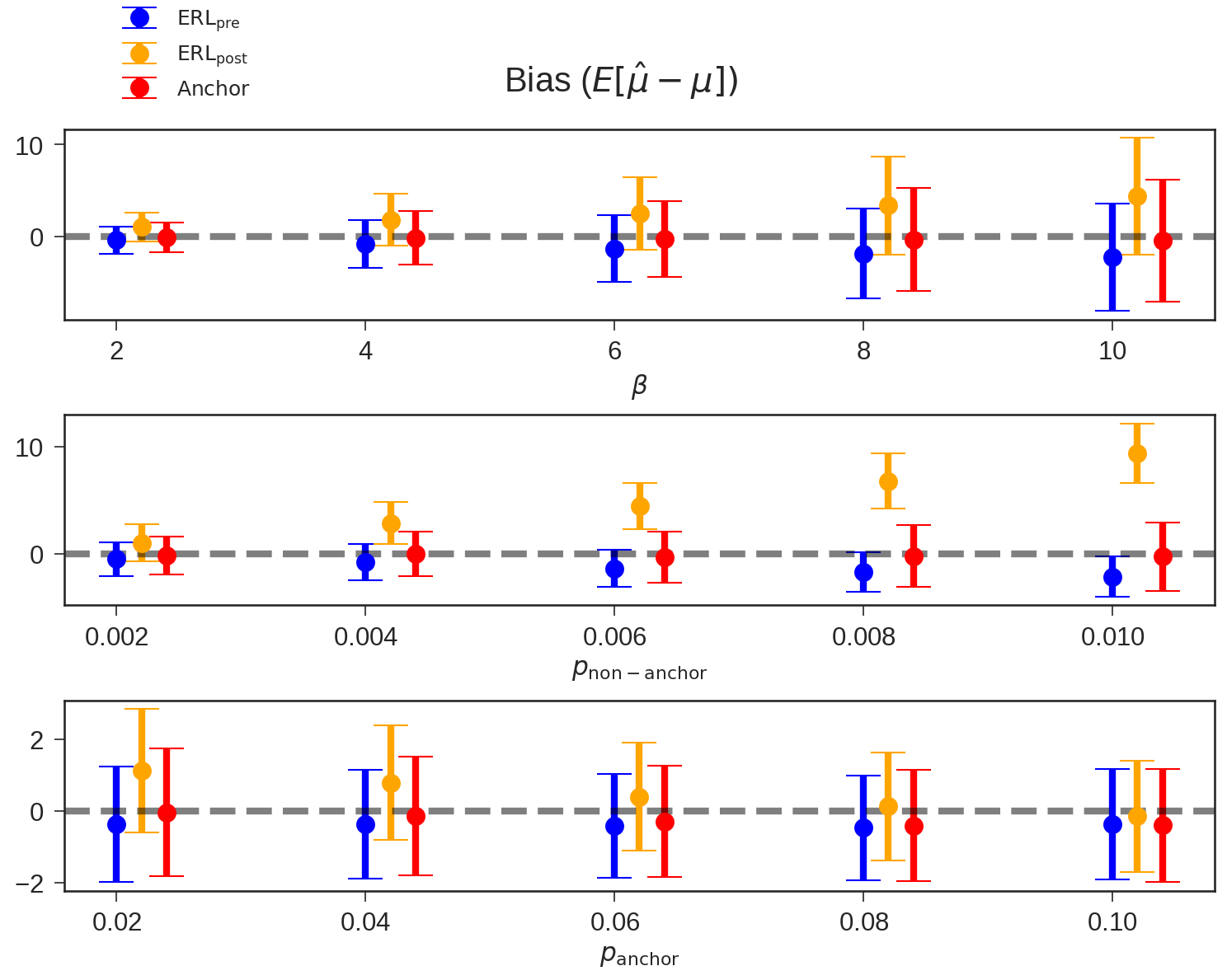}
  \caption{Bias with respect to different parameter values. Error bars are empirical standard deviations of the estimators over 1000 Monte Carlo runs.}
  \label{figure:bias}
\end{figure*}

\begin{figure*}[h]
  \centering
  \includegraphics[width=0.62\linewidth]{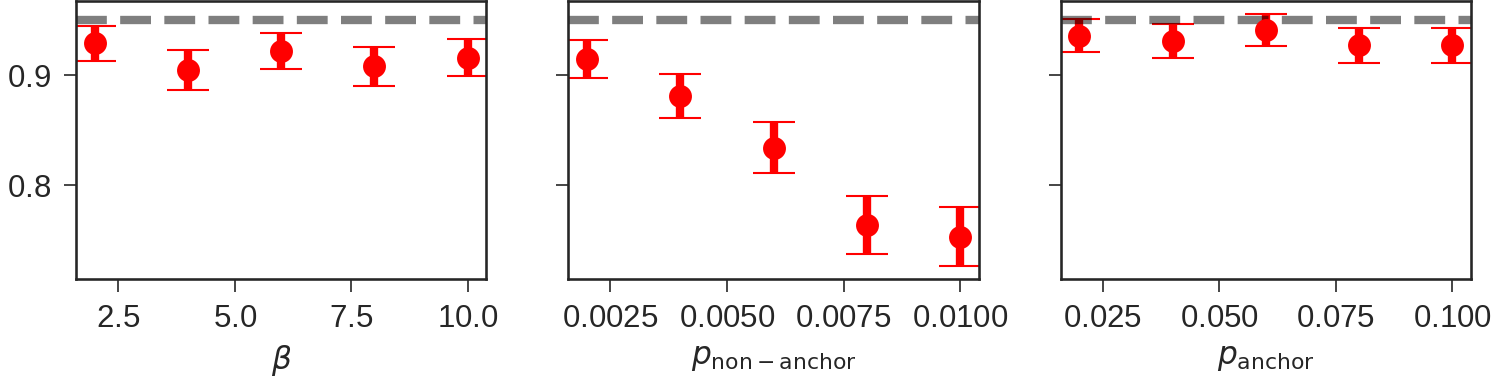}
  \caption{95\%-CI coverage from 1000 independent Monte Carlo runs with respect to different parameter values. Error bars represent the $95\%$-CI for the true coverage probability.}
  \label{figure:coverage}
\end{figure*}


\subsection{Uncertainty Quantification}
In order to provide a confidence interval (CI) for the anchor estimator, we appeal to asymptotic normality and use $\text{Var}(\tilde\mu^{u,c})$ under the sharp null as a proxy for the anchor estimator $\text{Var}(\hat\mu^{u,c})$, where $\tilde\mu^{u,c}$ is the estimator \eqref{equation:mu-tilde-for-test} defined in Section~\ref{subsec:testing}. We observe in Figure \ref{figure:coverage} that this approach delivers close to the nominal level coverage under different values of treatment effect ($\beta$) and anchor subgraph violation ($p_\text{anchor}$), with only modest under-coverage. By contrast, coverage is low when the edge formation density ($p_\text{non-anchor}$) is high, which is not surprising since the variance proxy is for an estimator that ignores the direct effect of newly formed edges.



\section{Endogenous Unipartite Graphs}
\label{sec:unipartite}
An endogenous unipartite interference graph is the special case of Definition~\ref{def:ebg}, where the analysis and randomization units are the same: $\mathcal{A} = \mathcal{R}$. This definition allows for \emph{directed} interference, in the sense that $E_{ar}$ may not equal $E_{ra}$. We specify the outcome model for all $a$ as $y_a = \alpha_a + \beta_a x_a + \gamma_a T_a$, where $x_a = \sum_r T_r e_{ar}w_{ar}$ and we impose $w_{aa} = 0$. Thus, $\beta_a x_a$ is the indirect effect on $a's$ outcome coming from the treatment of others, and $\gamma_a T_a$ is the direct effect coming from $a$ itself being treated. The TTE is $\frac{1}{n_a}\sum_a [\mathcal{W}_a(\mathbf{1})\beta_a + \gamma_a]$. Let $\hat\gamma_a = T_a y_a / p - (1 - T_a) y_a / (1 - p)$. The previous result on unbiasedness in the bipartite setting can be directly applied to construct a class of unbiased estimators in the unipartite setting. We employ the previously developed bipartite estimator to estimate the indirect effect and use $\hat\gamma_a$ as an unbiased estimator of the direct effect. Summing the two gives an unbiased estimator of the TTE, $\hat\mu^{u,c}_{uni}$, defined as
\begin{equation}
\label{equation:uni}
\hat\mu^{u,c}_{uni}=\frac{1}{n_a}\sum_a\hat\mu_{a,uni}^{u,c}=\frac{1}{n_a}\sum_a \left[ \hat\beta^u_a\cdot\widehat{\mathcal W_a(\mathbf 1)}^c + \hat\gamma_a \right].
\end{equation}

\begin{corollary} \label{corollary:unipartite-unbiasedness}
Under the assumptions of Theorem~\ref{theorem:unbiased}, with the outcome model $y_a = \alpha_a + \beta_a x_a + \gamma_a T_a$, if $w_{aa} = u_{aa} = 0$ for all $a$, an unbiased estimator of the TTE is $\hat\mu^{u,c}_{uni}$.
\end{corollary}

The restriction $w_{aa} = u_{aa} = 0$ is innocuous given the presence of a separate term $\gamma_a T_a$ in the outcome model for the direct effect. The summands in (\ref{equation:uni}) follow the same dependence structure as those in (\ref{equation:mu-hat}). Consistency and asymptotic normality consequently follow directly from the same arguments as the bipartite case, under the same conditions on the growth of the maximum degrees, $d_{\mathcal A}$ and $d_{\mathcal R}$.

\begin{corollary} \label{corollary:unipartite-consistent}
Under the assumptions of Theorem~\ref{theorem:consistent}, with the outcome model $y_a = \alpha_a + \beta_a x_a + \gamma_a T_a$, if $w_{aa} = u_{aa} = 0$ for all $a$, then \begin{enumerate}
    \item $\text{Var}(\hat\mu^{u,c}_{uni})=O(d_{\mathcal A}^3d_{\mathcal R}/n_a)$, and thus $\hat\mu^{u,c}$ is consistent if $d_{\mathcal A}^3d_{\mathcal R}/n_a\to0$;
    \item if $n_a\text{Var}(\hat\mu^{u,c}_{uni})$ is bounded away from zero and $d_{\mathcal A}^{10}d_{\mathcal R}^4/n_a\to0$, then ${(\hat\mu^{u,c}-\mu})/{\sqrt{\text{Var}(\hat\mu^{u,c}_{uni})}}$ converges in distribution to $\mathcal N(0,1)$.
\end{enumerate}
\end{corollary}

\section{Conclusion}
\label{sec:conclusion}
We establish a foundation for estimating and testing treatment effects on endogenous graphs, but significant extensions and open questions remain. Future research could characterize the variance-minimizing choice of weights $u_{ar}$ used to construct instruments $z^u_a$. A consistent and scalable variance estimator could be used in conjunction with Theorem~\ref{theorem:consistent} to construct confidence intervals around the TTE estimator, complementing the testing results of Section \ref{sec:tte_estimators}. And while we generalize the linear outcome model from the conventional exogenous edge setting, the outcome model may be implausible in some applications. We believe the research direction of causal inference on endogenous graphs holds promise for advancing our understanding of treatment effects in complex, real-world networks.

\section*{Acknowledgements}
The authors would like to thank Okke Schrijvers for his valuable help in this work.

\bibliography{refs}
\appendix
\section{Variance and Confidence Interval for ERL Estimators}
\label{sec:variance-estimate}
Let
\begin{equation*}
x_a^\text{pre}(\mathbf T)=\frac{\sum_{r\in\mathcal R}e_{ar}^\text{pre}T_r}{\sum_{r\in\mathcal R}e_{ar}^\text{pre}}
\end{equation*}
and
\begin{equation*}
x_a^\text{post}(\mathbf T)=\frac{\sum_{r\in\mathcal R}e_{ar}^\text{pre}T_r}{\sum_{r\in\mathcal R}e_{ar}^\text{post}},
\end{equation*}
where $e_{ar}^\text{pre}$ and $e_{ar}^\text{post}$ are the pre-treatment and post-treatment edge between $a$ and $r$. The ERL estimators \citep{harshaw2021design} are defined as
\begin{equation*}
\hat\mu^\text{pre}=\frac1{n_a}\sum_a\frac{y_a(x_a^\text{pre}(\mathbf T)-p)}{p(1-p)/(\sum_{r\in\mathcal R}e_{ar}^\text{pre})}
\end{equation*}
and
\begin{equation*}
\hat\mu^\text{post}=\frac1{n_a}\sum_a\frac{y_a(x_a^\text{pre}(\mathbf T)-p)}{p(1-p)/(\sum_{r\in\mathcal R}e_{ar}^\text{post})}
\end{equation*}
In stead of using the variance estimator proposed in \citet{harshaw2021design}, we calculate $\text{Var}(\hat\mu^\text{pre}\mid\mathbf y, \mathbf e^\text{pre})$ and $\text{Var}(\hat\mu^\text{post}\mid\mathbf y, \mathbf e^\text{post})$ under the strong null, where $(\mathbf y, \mathbf e^\text{pre}, \mathbf e^\text{post})$ is independent of $\mathbf T$, and use them to construct confidence intervals based on asymptotic normality in our simulations in Section~\ref{sec:simulation}.
\section{Proofs}
\label{sec:proofs}
\begin{proof}[Proof of Theorem 1]
We first show that $\textrm{Cov}(x_a, z^u_a)$ is known and non-zero so that $\hat\beta_a^u$ is a feasible estimator. By Lemma~\ref{lemma:cov}, $\textrm{Cov}(x_a, z^u_a)=p(1-p)\sum_rw_{ar}e_{ar}(1)u_{ar}=p(1-p)\sum_rw_{ar}u_{ar}$. The last simplification is justified because when $u_{ar}\ne0$, $e_{ar}(1)=1$. By assumptions (e) and (f) in the theorem statement, this covariance is non-zero.

Next, we need to show that $\textrm{Cov}\left(\hat\beta^u_a, \widehat{\mathcal W_a(\mathbf 1)}^c\right)=0$ for each $a\in\mathcal A$. This result is proved in Lemma~\ref{lemma:w-a-z}.

Now we know that $\mathbb E\left[\hat\beta^u_a \widehat{\mathcal W_a(\mathbf 1)}^c\right]=\mathbb E\left[\hat\beta^u_a\right]\mathbb E\left[\widehat{\mathcal W_a(\mathbf 1)}^c\right]=\beta_a\mathcal W_a(\mathbf 1)$. Averaging over $a\in\mathcal A$ proves the theorem.
\end{proof}

\begin{lemma}
In the setting of Definition 2, when the treatment assignments are pairwise independent with marginal probability $p$,
\begin{align*}
\text{Cov}\left(\sum_rs_{ar}(T_r)T_r, \sum_rv_{ar}(T_r)T_r\right) \\ 
=p(1-p)&\sum_rs_{ar}(1)v_{ar}(1),
\end{align*}
where $s_{ar}$, $v_{ar}$: $\{0,1\}\to\mathbb R$ are any deterministic functions.
\label{lemma:cov}
\end{lemma}
\begin{proof}[Proof of Lemma~\ref{lemma:cov}]
We notice that $T_rs_{ar}(T_r)\equiv T_rs_{ar}(1)$ and $T_rv_{ar}(T_r)\equiv T_rv_{ar}(1)$, so
\begin{align*}
&\quad\text{Cov}\left(\sum_rs_{ar}(T_r)T_r, \sum_rv_{ar}(T_r)T_r\right)\\
&= \text{Cov}\left(\sum_rs_{ar}(1)T_r, \sum_rv_{ar}(1)T_r\right)\\
  &= \sum_r\textrm{Cov}\left(s_{ar}(1)T_r, v_{ar}(1)T_{r}\right) \\& \quad{} +\sum_r \sum_{r'\ne r} \underbrace{\textrm{Cov}\left(s_{ar}(1)T_r, v_{ar'}(1)(T_{r'})\right)}_{=0, \text{ pairwise independence}}\\
  &= \sum_rs_{ar}(1)v_{ar}(1)\textrm{Cov}\left(T_r, T_{r}\right)\\
  &= p(1-p)\sum_rs_{ar}(1)v_{ar}(1).
\end{align*}
\end{proof}

\begin{lemma}
In the setting of Definition 2, suppose $G$ is an anchor sub-graph. If (a) edges are $r$-driven; (b) the treatment assignments are independent Bernoulli random variables with probability $p$; (c) $\{(a,r)\mid u_{ar}\ne0\}\subset G$; (d) $c_{ar}=\mathbb I((a,r)\in G)$; (e) $w_{ar}\ne0$ for all $(a,r)$; and (f) $|\{r\mid u_{ar}\ne0\}|>0$ for all $a$, then
\begin{equation*}
\textrm{Cov}\left(\widehat{\mathcal{W}_a(\mathbf{1})}^c,  \hat\beta_a^u\right)=0,
\end{equation*}
where $\hat\beta_a^u$ and $\widehat{\mathcal{W}_a(\mathbf{1})}^c$ are from equations (2) and (3).
\label{lemma:w-a-z}
\end{lemma}
\begin{proof}[Proof of Lemma~\ref{lemma:w-a-z}]
We omit the $u$ and $c$ superscripts.

We decompose the covariance as
\begin{multline*}
\textrm{Cov}\left(\widehat{\mathcal{W}_a(\mathbf{1})},  y_a\left(\frac{z_a - \mathbb{E}{z_a}}{\textrm{Cov}(x_a, z_a)}\right)\right)\\
=\underbrace{\mathbb E\left[\textrm{Cov}\left(\widehat{\mathcal{W}_a(\mathbf{1})},  y_a\left(\frac{z_a - \mathbb{E}{z_a}}{\textrm{Cov}(x_a, z_a)}\right)\mid T_A\right)\right]}_{X}\\
+\underbrace{\textrm{Cov} \left(\mathbb E\left[\widehat{\mathcal{W}_a(\mathbf{1})}\mid T_A\right],  \mathbb E\left[y_a\left(\frac{z_a - \mathbb{E}{z_a}}{\textrm{Cov}(x_a, z_a)}\right)\mid T_A\right]\right)}_{Y},
\end{multline*}
where $A=\{r\mid u_{ar}\ne0\}\subset\{r\mid (a,r)\in G\}$ is the set of anchor randomization units for analysis unit $a$, and $T_A=\{T_a\mid a\in A\}$. By definition, $z_a$ is a function of $T_A$.

For term $X$, note that
\begin{equation*}
\begin{aligned}
&\quad\textrm{Cov}\left(\widehat{\mathcal{W}_a(\mathbf{1})},  y_a\left(\frac{z_a - \mathbb{E}{z_a}}{\textrm{Cov}(x_a, z_a)}\right)\mid T_A\right)\\
&=\left(\frac{z_a - \mathbb{E}{z_a}}{\textrm{Cov}(x_a, z_a)}\right)\textrm{Cov}\left(\widehat{\mathcal{W}_a(\mathbf{1})},  y_a\mid T_A\right)\\
&=\left(\frac{z_a - \mathbb{E}{z_a}}{\textrm{Cov}(x_a, z_a)}\right)\underbrace{\beta_a(1-p)\sum_{r\not\in A}w_{ar}^2(e_{ar}(1)-c_{ar})e_{ar}(1)}_{\text{constant}}
\end{aligned}
\end{equation*}
is the product of a constant (by Lemma~\ref{lemma:cov}) and a mean-zero random variable, so $X$, its expectation, is zero.

For term $Y$, we have
\begin{equation*}
\begin{aligned}
\mathbb E\left[\widehat{\mathcal{W}_a(\mathbf{1})}\mid T_A\right]&=\underbrace{\sum_{r\in A}T_rw_{ar}(e_{ar}-c_{ar})/p}_{=0} \\& \quad{} +\sum_{r\in A}w_{ar}c_{ar}
+ \sum_{r\not\in A}w_{ar}e_{ar}(1),    
\end{aligned}
\end{equation*}
where the first term is zero because $T_rw_{ar}(e_{ar}-c_{ar})=T_rw_{ar}(e_{ar}(1)-c_{ar})$, and $e_{ar}(1)=c_{ar}=1$ for $r\in A$. This makes $\mathbb E\left[\widehat{\mathcal{W}_a(\mathbf{1})}\mid T_A\right]$ a constant, and thus $Y=0$.
\end{proof}

\begin{lemma}
Under the assumptions of Theorem 2, we have
\begin{equation}
|\hat\mu^{u,c}_a|=|\hat\beta^u_a\widehat{\mathcal W_a(\mathbf 1)}^c|\le\frac{MW_h|R_a|}{p^2(1-p)W_l}.
\label{equation:mu-hat-bound}
\end{equation}
\label{lemma:estimator-bound}
\end{lemma}

\begin{proof}[Proof of Lemma~\ref{lemma:estimator-bound}]
By the anchor sub-graph assumption, $V_a\subset R_a$. Then we have $x_a$ is a function of $T_{R_a}$ and $z_a^u$ is a function of $T_{V_a}$, and hence $\hat\beta^u_a$ is a function of $T_{R_a\cup V_a}=T_{R_a}$. For $\widehat{\mathcal W_a(\mathbf 1)}^c$, we notice that $c_{ar}=\mathbb I(r\in V_a)$, and then
\begin{equation}
\begin{aligned}
\widehat{\mathcal W_a(\mathbf 1)}^c&=\sum_r\left(\frac{T_r w_{ar}(e_{ar}(1) - c_{ar})}{\mathbb{P}(T_r=1)}+w_{ar}c_{ar}\right)\\
&=\sum_{r\in V_a}\left(\frac{T_r w_{ar}(e_{ar}(1) - 1)}{p}+w_{ar}\right) \\ 
& \quad{} +\sum_{r\not\in V_a}\left(\frac{T_r w_{ar}e_{ar}(1)}{p}\right)\\
&=\sum_{r\in V_a}w_{ar}+\sum_{r\not\in V_a}\left(\frac{T_r w_{ar}e_{ar}(1)}{p}\right)\\
&=\sum_{r\in V_a}w_{ar}+\sum_{r\in R_a\setminus V_a}\left(\frac{T_r w_{ar}e_{ar}(1)}{p}\right)\\& \quad{} +\sum_{r\not\in R_a}\left(\frac{T_r w_{ar}e_{ar}(1)}{p}\right)\\
&=\sum_{r\in V_a}w_{ar}+\sum_{r\in R_a\setminus V_a}\left(\frac{T_r w_{ar}}{p}\right)
\end{aligned}
\label{equation:w-hat-decomp}
\end{equation}
is a function of $T_{R_a\setminus V_a}$. As a result, $\hat\mu_a^{u,c}$ is a function of $T_{R_a}$. 

From \eqref{equation:w-hat-decomp}, we naturally have \begin{equation}
|\widehat{\mathcal W_a(\mathbf 1)}^c|\le \frac{W_h}{|R_a|}(|V_a|+(|R_a|-|V_a|)/p)\le\frac{W_h}{p}.
\label{equation:w-hat-bound}
\end{equation}

Next, we examine $|\hat\beta_a^u|$. Without loss of generality, we assume we have normalized $u$ so $\sum_ru_{ar}=1$. This is possible because we have assumed at least one $u_{ar}$ is non-zero for each $a$.
\begin{equation}
\begin{aligned}
|\hat\beta_a^u|&=|y_a||z_a^u-\mathbb E[z_a^u]|/|\text{Cov}(x_a,z_a^u)|\\
&\le M |\sum_ru_{ar}(T_r-p)|/|p(1-p)\sum_ru_{ar}w_{ar}e_{ar}(1)|\\
&\le M/|p(1-p)\sum_ru_{ar}w_{ar}e_{ar}(1)|\\
&= M/|p(1-p)\sum_{r\in V_a}u_{ar}w_{ar}|\\
&\le M|R_a|/|p(1-p)W_l\sum_{r\in V_a}u_{ar}|\\
&= \frac{M|R_a|}{p(1-p)W_l}.
\end{aligned}
\label{equation:beta-hat-bound}
\end{equation}

Combining \eqref{equation:w-hat-bound} and \eqref{equation:beta-hat-bound}, we arrive at \eqref{equation:mu-hat-bound}.
\end{proof}

\begin{proof}[Proof of Theorem 2] 

We define $\mathcal I(i)=\{j\in\mathcal R\mid R_i\cap R_j\ne\emptyset\}$. It follows directly that $|\mathcal I(i)|\le d_ad_r$. Then
\begin{equation}
\begin{aligned}
\text{Var}(\hat\mu)&=\frac{1}{n_a^2}\sum_{i\in\mathcal A}\sum_{j\in\mathcal A}\text{Cov}(\hat\mu_i,\hat\mu_j)\\
&=\frac{1}{n_a^2}\sum_{i\in\mathcal A}\sum_{j\in\mathcal I(i)}\text{Cov}(\hat\mu_i,\hat\mu_j)\\
&\quad\text{($\hat\mu_i$ is independent of $\hat\mu_j$ if $R_i\cap R_j=\emptyset$)}\\
&\le\frac{1}{n_a^2}\sum_{i\in\mathcal A}\sum_{j\in\mathcal I(i)}\sqrt{\text{Var}(\hat\mu_i)\text{Var}(\hat\mu_j)}\\
&\le\frac{1}{n_a^2}\sum_{i\in\mathcal A}\sum_{j\in\mathcal I(i)}\sqrt{\mathbb E(\hat\mu_i^2)\mathbb E(\hat\mu_j^2)}\\
&\le\frac{1}{n_a^2}\sum_{i\in\mathcal A}\sum_{j\in\mathcal I(i)}\frac{M^2W_h^2|R_i||R_j|}{p^4(1-p)^2W_l^2}\quad\text{(Lemma~\ref{lemma:estimator-bound})}\\
&\le\frac{d_ad_r}{n_a}\frac{M^2W_h^2d_a^2}{p^4(1-p)^2W_l^2}=O(d_a^3d_r/n_a).
\end{aligned}
\end{equation}
In the above, we have suppressed the $u,c$ superscripts for better readability. This proves claim (i).

We let $a_i=\hat\mu_{i}^{u,c}-\mu_i$ in Lemma~\ref{lemma:wasserstein}. We have shown in the proof of Theorem 1 that $\mathbb E[a_i]=0$. By Lemma~\ref{lemma:estimator-bound},
\begin{multline}
|a_i|\le |\hat\mu_{i}^{u,c}|+|\mu_i|
= |\hat\mu_{i}^{u,c}|+|\mathbb E[\hat\mu_{i}^{u,c}]|
\le |\hat\mu_{i}^{u,c}|+\mathbb E[|\hat\mu_{i}^{u,c}|]\\
\le \frac{2MW_h|R_a|}{p^2(1-p)W_l}\le\frac{2MW_hd_a}{p^2(1-p)W_l}.
\label{equation:a-i-bound}
\end{multline}
In the proof of Lemma~\ref{lemma:estimator-bound}, we have shown that both $\hat\beta$ and $\widehat{\mathcal W_a(\mathbf 1)}$ are functions of $T_{R_a}$. Therefore, $\hat\mu_a^{u,c}$ is independent of $\hat\mu_b^{u,c}$ if $R_a\cap R_b=\emptyset$. This means $D\le d_ad_r$. With this fact, and \eqref{equation:a-i-bound}, we have
\begin{align*}
&d_W\left(\frac{\hat\mu^{u,c}-\mu}{\sqrt{\text{Var}(\hat\mu^{u,c})}},Z\right)\\
&\le\frac{d_a^2d_r^2}{\sigma^3n_a^2}\left[\frac{2MW_hd_a}{p^2(1-p)W_l}\right]^3+\sqrt{\frac{28}\pi}\frac{d_a^{3/2}d_r^{3/2}}{\sigma^2n_a^{3/2}}\left[\frac{2MW_hd_a}{p^2(1-p)W_l}\right]^2\\
&=O\left(\frac{d_a^5d_r^2}{\sigma^3n_a^2}+\frac{d_a^{7/2}d_r^{3/2}}{\sigma^2n_a^{3/2}}\right)=O\left(\frac{d_a^5d_r^2}{n_a^{1/2}}+\frac{d_a^{7/2}d_r^{3/2}}{n_a^{1/2}}\right) \\ 
&=O\left(\frac{d_a^5d_r^2}{n_a^{1/2}}\right).
\end{align*}
This proves claim (ii).
\end{proof}

\begin{lemma}[Lemma~3.6 in \cite{ross2011fundamentals}]
Let $a_1,a_2,\dots,a_n$ be random variables with zero mean and finite fourth-moments. Define $X=\left(\frac1n\sum_{i=1}^na_i\right)/\sigma$, where $\sigma^2=\text{Var}\left(\frac1n\sum_{i=1}^na_i\right)$. Then for a standard normal random variable $Z\sim\mathcal N(0,1)$, we have
\begin{equation}
d_W(X,Z)\le\frac{D^2}{\sigma^3n^3}\sum_{i=1}^n\mathbb E[|a_i|^3]+\sqrt{\frac{28}\pi}\frac{D^{3/2}}{n^2\sigma^2}\sqrt{\sum_{i=1}^n\mathbb E[a_i^4]},
\end{equation}
where $D$ is the maximum dependency degree of the random variables and $d_W (\cdot, \cdot)$ is the Wasserstein distance.
\label{lemma:wasserstein}
\end{lemma}

\begin{proof}[Proof of Theorem 3] 
Due to the discrete nature of $\tilde\mu^{u,c}$, randomization may be needed to achieve exact level-$\alpha$. This concern goes away in the limit.

We suppress the superscript $(u,c)$ in $\tilde\mu$.  The validity of test follows by definition. We begin by showing that $c_\alpha\to0$ and $\tilde\mu-\mu^*\overset{p}\to0$, where $\mu^*=\frac1{n_a}\sum_a\beta_a\sum_r(w_{ar}c_{ar})$.

We first show $c_\alpha\to0$. We define $\tilde\mu_i=\hat\beta_i\sum_r(w_{ir}c_{ir})$, and $\mathcal J(i)=\{j\in\mathcal R\mid V_i\cap V_j\ne\emptyset\}$. It follows directly that $|\mathcal J(i)|\le d_ad_r$. From \eqref{equation:beta-hat-bound}, we have
\begin{equation}
|\tilde\mu_i|\le|\hat\beta_i||\sum_r(w_{ir}c_{ir})|\le\frac{M|R_a|}{p(1-p)W_l}\cdot \frac{W_h}{|R_a|}|V_a|.
\label{equation:tilde-mu-bound}
\end{equation}
Under $H_0^\text{sharp}$, the conditional variance of $\tilde\mu$ given $\{y_a\mid a\in\mathcal A\}$ is
\begin{equation}
\begin{aligned}
\text{Var}(\tilde\mu\mid y)&=\frac{1}{n_a^2}\sum_{i\in\mathcal A}\sum_{j\in\mathcal A}\text{Cov}(\tilde\mu_i,\tilde\mu_j\mid y)\\
&=\frac{1}{n_a^2}\sum_{i\in\mathcal A}\sum_{j\in\mathcal J(i)}\text{Cov}(\tilde\mu_i,\tilde\mu_j\mid y)\\
&\quad\text{(if $V_i\cap V_j=\emptyset$, $z_i^u$ is independent of $z_j^u$}\\
&\quad\text{and thus $\tilde\mu_i$ is independent of $\tilde\mu_j$)}\\
&\le\frac{1}{n_a^2}\sum_{i\in\mathcal A}\sum_{j\in\mathcal J(i)}\sqrt{\text{Var}(\tilde\mu_i\mid y)\text{Var}(\tilde\mu_j\mid y)}\\
&\le\frac{1}{n_a^2}\sum_{i\in\mathcal A}\sum_{j\in\mathcal J(i)}\sqrt{\mathbb E(\tilde\mu_i^2)\mathbb E(\tilde\mu_j^2)}\\
&\le\frac{1}{n_a^2}\sum_{i\in\mathcal A}\sum_{j\in\mathcal J(i)}\frac{M^2W_h^2|V_i||V_j|}{p^2(1-p)^2W_l^2}\quad\text{(because of \eqref{equation:tilde-mu-bound})}\\
&\le\frac{d_ad_r}{n_a}\frac{M^2W_h^2d_a^2}{p^2(1-p)^2W_l^2}=O(d_a^3d_r/n_a)\to0.
\end{aligned}
\label{equation:tilde-mu-cond-var}
\end{equation}
Note that under $H_0^\text{sharp}$,
\begin{equation*}
\begin{aligned}
\text{Var}(\tilde\mu\mid y)&=\mathbb E[\tilde\mu^2\mid y]\\
&= \mathbb E[\tilde\mu^2\mid y, |\tilde\mu|\ge c_\alpha]P(|\tilde\mu|\ge c_\alpha\mid y)\\
&\quad+\mathbb E[\tilde\mu^2\mid y, |\tilde\mu|< c_\alpha]P(|\tilde\mu|< c_\alpha\mid y)\\
&\ge\mathbb E[\tilde\mu^2\mid y, |\tilde\mu|\ge c_\alpha]P(|\tilde\mu|\ge c_\alpha\mid y)\\
&\ge c_\alpha^2\alpha.
\end{aligned}
\end{equation*}
We proved the left hand side goes to zero, so $c_\alpha$ must go to zero.

Next, we show $\tilde\mu-\mu^*\overset{p}{\to}0$. From the fact that $\hat\beta_a$ is unbiased for $\beta_a$, we get
\begin{equation*}
\mathbb E[\tilde\mu]=\frac1{n_a}\sum_a\beta_a\sum_r(w_{ar}c_{ar})\to\mu^*.
\end{equation*}
As for $\text{Var}(\tilde\mu)$,
\begin{equation}
\begin{aligned}
\text{Var}(\tilde\mu)&=\frac{1}{n_a^2}\sum_{i\in\mathcal A}\sum_{j\in\mathcal A}\text{Cov}(\tilde\mu_i,\tilde\mu_j)\\
&=\frac{1}{n_a^2}\sum_{i\in\mathcal A}\sum_{j\in\mathcal I(i)}\text{Cov}(\tilde\mu_i,\tilde\mu_j)\\
&\quad\text{(if $R_i\cap R_j=\emptyset$, $\tilde\mu_i$ is independent of $\tilde\mu_j$)}\\
&\le\frac{1}{n_a^2}\sum_{i\in\mathcal A}\sum_{j\in\mathcal I(i)}\sqrt{\text{Var}(\tilde\mu_i)\text{Var}(\tilde\mu_j)}\\
&\le\frac{1}{n_a^2}\sum_{i\in\mathcal A}\sum_{j\in\mathcal I(i)}\sqrt{\mathbb E(\tilde\mu_i^2)\mathbb E(\tilde\mu_j^2)}\\
&\le\frac{1}{n_a^2}\sum_{i\in\mathcal A}\sum_{j\in\mathcal I(i)}\frac{M^2W_h^2|V_i||V_j|}{p^2(1-p)^2W_l^2}\quad\text{(because of \eqref{equation:tilde-mu-bound})}\\
&\le\frac{d_ad_r}{n_a}\frac{M^2W_h^2d_a^2}{p^2(1-p)^2W_l^2}=O(d_a^3d_r/n_a)\to0,
\end{aligned}
\end{equation}
where $\mathcal I(i)=\{j\in\mathcal R\mid R_i\cap R_j\ne\emptyset\}$ is defined at the beginning of the proof of Theorem 2. Now we have shown $\mathbb E[\tilde\mu-\mu^*]=0$ and $\text{Var}(\tilde\mu-\mu^*)\to0$, which means $\tilde\mu-\mu^*\overset{p}{\to}0$.

Finally, there exists $\varepsilon>0$ so that eventually $|\mu^*|>\varepsilon$ (because the limit inferior of $|\mu^*|$ is positive) and $c_\alpha<\varepsilon/2$ (because $c_\alpha\to0$) as $n_a\to\infty$. Once this happens, 
\begin{equation*}
P(|\tilde\mu|\le c_\alpha)\le P(|\tilde\mu|\le \varepsilon/2)\le P(|\tilde\mu-\mu^*|\ge \varepsilon/2)\to1.
\end{equation*}
\end{proof}

\end{document}